\documentclass[conference,10pt]{IEEEtran}
\IEEEoverridecommandlockouts

\usepackage{cite}
\usepackage{amsmath,amssymb,amsfonts,upgreek}
\usepackage{algorithmic}
\usepackage{textcomp}
\usepackage{xcolor}
\usepackage{cite}
\usepackage{url}
\usepackage{amsfonts}
\usepackage{amssymb}
\usepackage{mathrsfs}
\usepackage{euscript}
\usepackage{graphicx}
\usepackage{multirow}
\usepackage{algorithm}
\usepackage{algorithmic}
\usepackage{changepage}
\usepackage{color}
\usepackage{caption}
\usepackage{subcaption}
\usepackage{todonotes}
\usepackage{scalerel}
\usepackage{tikz}
\usetikzlibrary{svg.path}
\usepackage{balance}
\usepackage{svg}
\usetikzlibrary{fit, positioning, shapes.geometric, fit, arrows.meta, calc, backgrounds}
\usepackage{pgfplots}
\pgfplotsset{compat=newest}
 
\definecolor{darkgreen}{rgb}{0, 0.5, 0} 
\definecolor{lightpurple}{rgb}{0.7, 0.4, 1} 

\definecolor{orcidlogocol}{HTML}{A6CE39}
\tikzset{
orcidlogo/.pic={
\fill[orcidlogocol] svg{M256,128c0,70.7-57.3,128-128,128C57.3,256,0,198.7,0,128C0,57.3,57.3,0,128,0C198.7,0,256,57.3,256,128z};
\fill[white] svg{M86.3,186.2H70.9V79.1h15.4v48.4V186.2z}
svg{M108.9,79.1h41.6c39.6,0,57,28.3,57,53.6c0,27.5-21.5,53.6-56.8,53.6h-41.8V79.1z M124.3,172.4h24.5c34.9,0,42.9-26.5,42.9-39.7c0-21.5-13.7-39.7-43.7-39.7h-23.7V172.4z}
svg{M88.7,56.8c0,5.5-4.5,10.1-10.1,10.1c-5.6,0-10.1-4.6-10.1-10.1c0-5.6,4.5-10.1,10.1-10.1C84.2,46.7,88.7,51.3,88.7,56.8z};
}
}
\newcommand\orcidicon[1]{\href{https://orcid.org/#1}{\mbox{\scalerel*{
\begin{tikzpicture}[yscale=-1,transform shape]
\pic{orcidlogo};
\end{tikzpicture}
}{|}}}}

\usepackage[colorlinks=true,linkcolor=blue,citecolor=blue]{hyperref}
\newtheorem{proof}{Proof}
\newtheorem{proposition}{Proposition}

\usepackage[most]{tcolorbox}

\begin{document} 
\author{
$
\text{Jalal Jalali}^{\orcidicon{0000-0002-3609-6775}}\ \IEEEmembership{Member, IEEE}, 
\text{Rodrigo C. de Lamare}^{\orcidicon{0000-0003-2322-6451}}\ \IEEEmembership{Senior Member, IEEE}
$.
}

\author{
\IEEEauthorblockN{
~Jalal Jalali\IEEEauthorrefmark{2},
~Mostafa Darabi\IEEEauthorrefmark{4}
\IEEEauthorrefmark{5},
and
Rodrigo C. de Lamare\IEEEauthorrefmark{2}\IEEEauthorrefmark{3}}
\IEEEauthorblockA{\IEEEauthorrefmark{2}A Centre for Telecommunications Studies, Pontifical Catholic University of Rio de Janeiro, Brazil}
\IEEEauthorblockA{\IEEEauthorrefmark{4}Department of ECE, University of British Columbia, Vancouver, BC V6T 1Z4, Canada\\}
\IEEEauthorblockA{\IEEEauthorrefmark{5}Wireless Communication Research Group, JuliaSpace Inc., Chicago, IL, USA\\}
\IEEEauthorblockA{\IEEEauthorrefmark{3} University of York, United Kingdom\\}
Emails: 
\href{mailto:josh@juliaspace.com}
{\texttt{josh@juliaspace.com}},
\href{mailto:mostafadarabi@ece.ubc.ca}
{\texttt{mostafadarabi@ece.ubc.ca}},
\href{mailto:delamare@puc-rio.br}
{\texttt{delamare@puc-rio.br}}
\vspace{-4mm}
}


\title{\huge Shape Adaptive Reconfigurable Holographic Surfaces}

\maketitle
 
\begin{abstract}
Reconfigurable Intelligent Surfaces (RIS) have emerged as a key solution to dynamically adjust wireless propagation by tuning the reflection coefficients of large arrays of passive elements. Reconfigurable Holographic Surfaces (RHS) build on the same foundation as RIS but extend it by employing holographic principles for finer-grained wave manipulation | that is, applying higher spatial control over the reflected signals for more precise beam steering. In this paper, we investigate shape-adaptive RHS deployments in a multi-user network. Rather than treating each RHS as a uniform reflecting surface, we propose a selective element activation strategy that dynamically adapts the spatial arrangement of deployed RHS regions to a subset of predefined shapes. In particular, we formulate a system throughput maximization problem that optimizes the shape of the selected RHS elements, active beamforming at the access point (AP), and passive beamforming at the RHS to enhance coverage and mitigate signal blockage. The resulting problem is non-convex and becomes even more challenging to solve as the number of RHS and users increases; to tackle this, we introduce an alternating optimization (AO) approach that efficiently finds near-optimal solutions irrespective of the number or spatial configuration of RHS. Numerical results demonstrate that shape adaptation enables more efficient resource distribution, enhancing the effectiveness of multi-RHS deployments as the network scales.
\end{abstract}

\vspace{-2mm}

\section{Introduction}
\vspace{-3mm}
\indent 

{In the evolution towards 6th generation (6G) networks, achieving high data rates remains a fundamental requirement for next-generation cellular systems. However, traditional approaches to meeting this demand often involve increased hardware complexity and energy consumption, particularly in systems that rely on large-scale antenna arrays for high-gain beamforming~\cite{6736746}. This has motivated the search for alternative, more efficient solutions that can enhance wireless performance without significantly increasing infrastructure costs.}

{Among these solutions, Reconfigurable Intelligent Surfaces (RIS) have gained significant attention in both academia and industry as a promising approach to improving wireless network efficiency in a cost-effective manner \cite{9371019, Jalali2024IRS}. RIS technology enables the dynamic reconfiguration of the wireless propagation environment, enhancing communication capabilities \cite{9326394}. Unlike traditional relays, RIS consists of a large array of low-cost passive reflecting elements that can intelligently adjust signal phases without requiring power-intensive radio frequency (RF) chains \cite{Lemic2025RIS}. By dynamically altering the reflection pattern, RIS can not only mitigate signal degradation but also optimize signal propagation paths to improve link reliability and throughput.}


To fully unlock the potential of RIS-assisted wireless networks, researchers have explored various optimization strategies for reflection coefficient tuning, beamforming, and RIS deployment. 
For instance, Souza \emph{et al.}~\cite{10639083} studied energy efficiency (EE) maximization in multi-user (MU) massive multiple-input multiple-output (mMIMO) systems by jointly optimizing base station antennas and RIS elements. Their reinforcement-learning-based method showed a 40.3\% improvement in EE over traditional optimization approaches. Similarly, Hou \emph{et al.}~\cite{10139787} investigated joint RIS selection and beamforming in multi-RIS unmanned aerial vehicle (UAV)-assisted anti-jamming networks, proposing a distributed matching-based selection and Q-learning-based beamforming to enhance network robustness.

Other works have focused on active element selection in RIS panels to improve transmission efficiency. Mu \emph{et al.}~\cite{10373958} introduced an active element selection (AES) algorithm for hybrid RIS-assisted device-to-device (D2D) networks, where a subset of RIS elements is activated for power amplification while others remain passive. Their optimization strategy significantly improved throughput performance while reducing energy consumption. In the context of multi-RIS-aided UAV networks, Bansal \emph{et al.}~\cite{10129204} addressed RIS selection under imperfect and outdated channel state information (CSI). Their work derived statistical models for selection probability, outage probability, and bit error rate (BER), showing that selecting fewer, larger RIS can outperform selecting many smaller ones.

\begin{figure*}
\vspace{-5mm}
    \centering
    \includegraphics[width=0.99\linewidth]{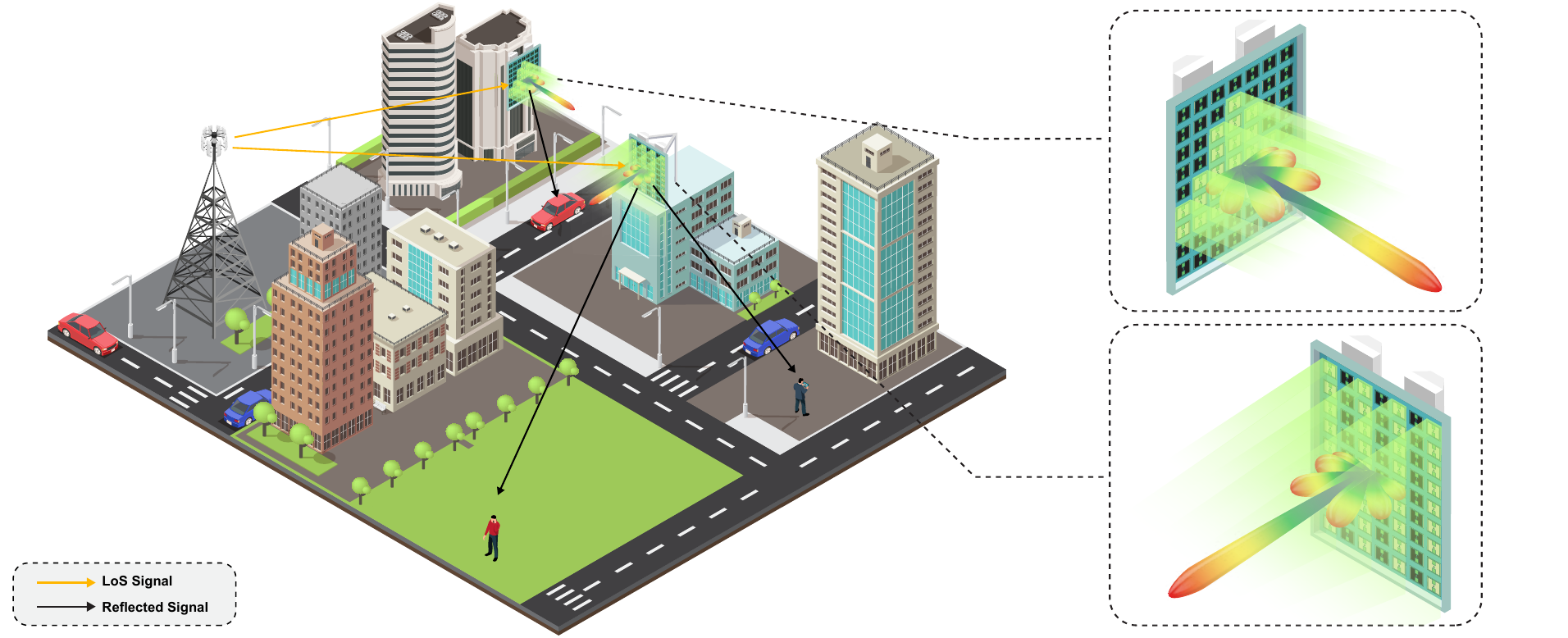}
    \vspace{-1mm}
    \caption{Shape-adaptive multi-RHS-assisted network architecture. The figure illustrates a downlink communication system where an AP, equipped with a ULA, serves multiple UEs through multiple RHS. Each RHS is composed of a planar array of elements, with different activation patterns optimized for network performance. The left RHS follows a rectangular activation shape, while the middle RHS utilizes a column-wise striped activation shape. These shape variations influence the reflective beamforming and impact signal propagation. Elements that do not constitute a given shape within an RHS can be repurposed for other functionalists, e.g., energy harvesting or sensing, while the utilized elements remain dedicated to communication~\cite{song2024miniature}.}
    \label{fig:sysmodel}
    \vspace{-5mm}
\end{figure*}

Beyond traditional network optimization, Duan \emph{et al.}~\cite{10433701} explored the role of active RIS-assisted power and element allocation in enhancing physical layer security (PLS) for AI-generated content (AIGC). Their method effectively mitigates the multiplicative fading effect of passive RIS systems, improving secrecy capacity through optimized power allocation. Additionally, Fang \emph{et al.}~\cite{9837009} proposed location-based RIS selection strategies for maximizing signal-to-noise ratio (SNR), demonstrating that optimized RIS selection policies significantly improve network performance. Moreover, Jalali \emph{et al.}~\cite{10872813} investigated placement, orientation, and resource allocation optimization for cell-free optical RIS (ORIS)-aided optical wireless communication (OWC) networks, using a multi-objective optimization problem (MOOP) approach to jointly improve spectral efficiency (SE) and EE while accounting for ORIS placement and user orientation. 

While prior research has extensively investigated RIS reflection coefficient optimization, active/passive element selection, and placement strategies~\cite{10639083,10139787,10373958,10129204,10433701,9837009,10872813}, these works have generally treated RIS panels as static structures. An underexplored yet critical component of RIS design is \textit{shape adaptation}, where the geometry of the active elements can be dynamically reconfigured. Existing approaches typically focus on selecting which elements to activate on a single RIS panel or choosing which RIS to deploy in multi-RIS scenarios. However, such methods overlook the potential performance gains from adjusting the shape of the deployed elements.

{In this work, we introduce the concept of \textit{shape adaptation using Reconfigurable Holographic Surfaces (RHS)} | a next-generation metasurface technology that extends beyond conventional RIS. Unlike standard RIS, which operate by reflecting signals from external transmitters, RHS incorporate an integrated feed directly into the metasurface. This design exploits a leaky-wave mechanism to radiate signals, rather than merely reflecting them, leading to a thinner form factor and more flexible beam synthesis capabilities~\cite{song2024miniature}. Because each element of an RHS can be selectively activated or adapted in shape, it can dynamically tailor both the amplitude and distribution of outgoing signals. In contrast, conventional RIS merely apply phase shifts to an incident wave from an external feed.}
We aim to maximize the downlink throughput in a multi-RHS-assisted network by exploiting shape-adaptive RHS deployment. The considered problem is non-linear and non-convex, becoming increasingly challenging to solve as the number of RHS and users grows. To address this, we propose an alternating optimization (AO) approach that efficiently finds near-optimal solutions regardless of the number or spatial configuration of RHS. To achieve this, we decompose the problem into three subproblems. The first subproblem is solved by searching for the optimal channel gain across all predefined RHS shapes, while the second and third subproblems optimize the active beamforming at the access point (AP) and the passive beamforming at the RHS, respectively.

The main contributions of our work are:
\begin{itemize}  
\item We propose a shape-adaptive multi-RHS deployment that dynamically selects utilized RHS regions to optimize signal propagation and mitigate blockage effects. 
\item Based on the proposed shape-adaptive RHS framework, we formulate a network throughput optimization problem that integrates shape adaptation, active beamforming at AP, and passive beamforming at RHS.   
\item To efficiently solve this inherently non-convex problem, we develop an AO approach, where the original problem is decomposed into three subproblems. 
\item To validate our approach, we conduct simulations, demonstrating that shape-adaptive RHS deployments significantly improve throughput performance. 
\end{itemize}

\textit{Notations:}
Vectors and matrices are denoted by boldface lowercase (e.g., $\mathbf{x}$) and uppercase (e.g., $\mathbf{X}$) letters, respectively.
The operators $(\cdot)^{*}$, $(\cdot)^{\mathrm{T}}$, and $(\cdot)^{\mathrm{H}}$ represent the conjugate, transpose, and Hermitian transpose, respectively.
Hadamard product is written as $\mathbf{A} \circ \mathbf{B}$.
The functions $\mathrm{tr}(\cdot)$, $\mathrm{vec}(\cdot)$, and $\mathrm{diag}(\cdot)$ denote the trace, vectorization, and diagonalization operators, while $\Re(\cdot)$ denotes the real part of a complex value.

\section{System Model}
Figure~1 illustrates the downlink of an RHS-assisted communication system. In this setup, an AP, equipped with a uniform linear array (ULA) consisting of $N_{\textrm{tr}}$ antenna elements, transmits signals to user equipment (UE) through multiple RHS panels. Each RHS employs a distinct activation shape to optimize network performance.
The system features $S$ RHS units, each tasked with serving $K$ single antenna UES. The set of RHS and UEs are denoted by $\mathcal{S} =\{1,\dots,S\}$ and $\mathcal{K} =\{1,\dots,K\}$, respectively. 
The AP’s communication is enhanced through passive RHS, each containing $M$ reflecting elements, where a controller adjusts the phase shifts.
Due to environmental obstacles, the line-of-sight (LoS) link between the AP and the UEs is assumed to be blocked, meaning that communication occurs exclusively through RHS-assisted signal reflection. To overcome these challenges, we aim to \textit{dynamically adapt the shape of each RHS} by selectively activating specific reflecting elements. This adaptation optimizes signal propagation and maximizes the overall network data rate.
Accordingly, for the $k$-th UE, the received signal can be formulated as follows:
\begin{align} 
{y_k} &= \nonumber
\sum\limits_{s \in \mathcal{S}} \mathbf{g}_{s,k}^{\textrm{H}} 
\left(\boldsymbol{A} \circ
\mathbf{\Phi}_s^{\textrm{H}}\right)
\mathbf{H}_s
\mathbf{w}_k
b_k 
\\&+
\underbrace{
\sum\limits_{k' \in \mathcal{K} k'\neq k } 
\sum\limits_{s \in \mathcal{S}} \mathbf{g}_{s,k'}^{\textrm{H}} 
\left(\boldsymbol{A} \circ
\mathbf{\Phi}_s^{\textrm{H}}\right)
\mathbf{H}_s
\mathbf{w}_{k'}
b_{k'}
}_{\text{interference$\triangleq i_k$}}
+ 
\underbrace{{n_k}}_{\text{noise}}, 
\end{align}
where $\mathbf{H}_s \!\in\! \mathbb{C}^{M \times N_{\textrm{tr}}}$ represents the channel between the AP and the $s^{\text{th}}$ RHS unit and $\mathbf{g}_{s,k} \!\in\! \mathbb{C}^{M \times 1}$ is the channel connecting the $s^{\text{th}}$ RHS unit to the $k^{\text{th}}$ UE. 
The matrix $\boldsymbol{A} \!\!\in\!\! \mathcal{A} \!\subset\!\! \mathbb{C}^{M \times M}$ is a binary diagonal shaping matrix that dynamically adapts based on the channel to optimize RHS-assisted communication. Unlike a fixed structure, $\boldsymbol{A}$ is selected from a predefined set $\mathcal{A}$, which includes specific RHS configurations, e.g., rectangular, strip-line or circular shapes. 
The phase shift matrix corresponding to the $s^{\text{th}}$ RHS unit is given by $
\boldsymbol{\Phi}_s \!=\! \eta^{-\frac{1}{2}} \textrm{diag}([e^{j\phi_{s,1}}, \dots, e^{j\phi_{s,M}}]^T),$
where $\eta$ is the reflection coefficient and $\phi_{s,m} \in [0, 2\pi)$, $\forall s \in \mathcal{S},~ m \in \mathcal{M}=\{1,...,M\}$ is the phase shift associated with the $m^{\text{th}}$ element of the $s^{\text{th}}$ RHS.
For the $k^{\text{th}}$ UE, the transmitted signal from the AP $s_k$ has normalized power and zero mean. 
{For simplicity, the combined interference and noise term, $\xi_k\!=\!i_k\!+\!n_k$ is modeled as a circularly symmetric complex Gaussian random variable, i.e., $\mathcal{CN}(0,\! \sigma^2_{\xi})$~\cite{7051266}.}
Lastly, the active beamforming matrix in the multi-RHS network can be rewritten as $\mathbf{W} \!\!=\!\! [\mathbf{w}_1, \dots, \mathbf{w}_K]$, where each precoding vector $\mathbf{w}_k \!\!\in\! \mathbb{C}^{N \times 1}$ is used by the AP for transmitting the signal~$b_k$. 

\section{Problem Formulation and Transformation}
Unlike conventional approaches that assume a fully utilized RHS surface, our method selectively activates RHS elements to form optimized reflection patterns based on a subset of predefined shapes. This selective activation enhances signal propagation and coverage while reducing unnecessary energy consumption. Furthermore, instead of focusing solely on beamforming optimization, we introduce shape adaptation as a key design variable to improve spectral efficiency and network performance\footnote{While throughput maximization is our primary objective, extending the formulation to incorporate interference-aware constraints and SINR optimization could provide a more comprehensive evaluation of system performance, particularly in multi-user scenarios with significant co-channel interference.}
.
The signal-to-noise ratio (SNR) experienced by the $k$-th UE can be expressed as follows:
\begin{equation}{\gamma_{k}} = 
{| {\sum\limits_{s \in \mathcal{S}} {\mathbf{g}_{s,k}^{\textrm{H}}} 
\left(\boldsymbol{A} \circ\mathbf{\Phi}_s^{\textrm{H}}\right)
{{\mathbf{H}_s}}{{\mathbf{w}}_k}} |}^2 \big/
\sigma_{\xi}^2,
\forall k.
\label{SNIR}
\end{equation}
The SNR expression quantifies the received signal quality at each UE, influencing achievable data rates. Based on this, we formulate the throughput optimization problem as follows:
\begin{subequations} 
\begin{align} 
\text{P}_1&: 
\mathop {\max}
\limits_{\boldsymbol{A},\mathbf{W},\mathbf{\Phi}_s} 
\sum\limits_{k = 1}^K 
\omega_k 
{\log_2}
\left( {1 + {\gamma_{k}}} \right)
\nonumber
\\ 
s.t.&:~
\quad \operatorname{tr}
\left( {{\mathbf{W}}{{\mathbf{W}}^{\text{H}}}} \right) \leq p_{\max},
\label{p1_c1}
\\ 
&\quad\quad
\big|
\left(\boldsymbol{A} \circ\mathbf{\Phi}_s^{\textrm{H}}\right)
\mathbf{H}_s
\mathbf{g}_{s,k}^{\textrm{H}}
\mathbf{w}_k
\big|^2
\geq p_{s}^{thr},
\forall s, k,
\label{p1_c2}
\\ 
&\quad\quad 
\theta_{s,m}
\in 
\{\theta_{s,m} = e^{j\phi_{s,m}}|\phi_{s,m} \in [0, 2\pi)\},
\forall s, m,
\label{p1_c3}\\
&\quad\quad 
\boldsymbol{A}\in \mathcal{A},
\label{p1_c4}
\end{align}
\end{subequations} 
where $\omega_k$ is the weight assigned to the data rate of the $k$-th UE.
Constraint \eqref{p1_c1} ensures that the total transmit power at the AP does not exceed the maximum allowed power $p_{\max}$.
Constraint \eqref{p1_c2} guarantees that each UE receives a minimum signal power $p_{s}^{thr}$ from all RHS, ensuring uniform signal reception.
Constraint \eqref{p1_c3} restricts the phase shifts $\phi_{s,m}$ to continuous values.
Finally, constraint \eqref{p1_c4} ensures that $\boldsymbol{A}$ is chosen from a predefined set of RHS shape configurations optimized for signal propagation based on channel conditions.

\begin{proposition}\label{propos_1}
The objective function of the optimization problem $\text{P}_1$ can be reformulated into an equivalent form:
\begin{subequations} 
\begin{align} 
\text{P}_2&: 
\!\!\!
\mathop {\max}
\limits_{\boldsymbol{A},\mathbf{W},\boldsymbol{\Phi}_s,{\boldsymbol{\chi}}} 
\frac{1}{{\ln 2}}\!
\sum\limits_{k = 1}^K 
\omega_k 
\ln 
\left( 1 \!+\! \chi_k \right) \!-\! 
\omega_k\chi_k \!+\! 
\frac{
\omega_k
\left(1 + \chi_k \right)
\gamma _k}
{1 + \gamma_k}
\nonumber
\\
s.t.&:~
\quad 
\eqref{p1_c1}-\eqref{p1_c4}.
\end{align}
\end{subequations} 
Here, $\boldsymbol{\chi} = [\chi_1, \dots, \chi_K]^T$ is an auxiliary vector introduced via fractional transformation~\cite{10464825,9123410}. 
\end{proposition}
\begin{proof}
Define the objective function of $\text{P}_2$ as $\Delta(\mathbf{W},\boldsymbol{\Phi}_s,\boldsymbol{\chi})$. This function is concave and differentiable with respect to $\chi_k$ when $\gamma_k$ is known. 
Taking the first-order derivative
$\frac{\partial \Delta(\mathbf{W},\boldsymbol{\Phi}_s,\boldsymbol{\chi})}
{\partial \chi_k} = 0$ yields the optimal value
$\hat{\chi}^k = \gamma_k$. 
By substituting $\hat{\boldsymbol{\chi}}$ into the objective function of $\text{P}_1$, we transform it into $\text{P}_2$, proving that both problems share the same optimal solution, establishing their equivalence. 
\hfill $\blacksquare$
\end{proof}
Building on \textit{proposition}~\ref{propos_1}, we propose an AO approach to efficiently solve the equivalent form of $\text{P}_1$. 
Since $\boldsymbol{\chi}$ can be directly obtained, its computation is straightforward.
Given $\boldsymbol{A}$, $\mathbf{W}$ and $\boldsymbol{\Phi}_s$,  the optimal values for $\chi_k$ can be determined using the first-order derivative of the objective function in $\text{P}_2$.
Thus, 
${\chi}_k^{(i)} = \gamma_k$, where $i$ is the iteration index. 
To tackle the remaining optimization variables, we introduce a three-step iterative algorithm, allowing independent tuning of RHS shape adaptation, AP beamforming, and RHS phase shifts.

\subsection{Step-one: Shape Adaptation}
\vspace{-0mm}
We first consider the shape adaptation for the RHS configuration in the optimization problem $\text{P}_2$. 
To maximize received signal power at the UE while satisfying design constraints, the RHS shape should be dynamically selected based on channel conditions. This selection/shape adaptation is formulated as: 
\begin{equation}
    \boldsymbol{A}=
    \underset{\boldsymbol{A} \in \mathcal{A}}{\text{argmax}}~~\sum\limits_{s \in \mathcal{S}} \sum\limits_{k \in \mathcal{K}} |\mathbf{g}_{s,k}^{\textrm{H}} 
    \left(\boldsymbol{A} \circ
    \mathbf{\Phi}_s^{\textrm{H}}\right)
    \mathbf{H}_s
    \mathbf{w}_k|^2.
    \label{shape_select}
\end{equation}
This shape selection strategy evaluates different RHS configurations based on effective channel gain, \eqref{shape_select}, and selects the one that maximizes received signal strength. By dynamically adapting the RHS shape, we achieve higher beamforming gain, broader spatial coverage, and improved interference suppression.
Moreover, since the number of available shapes in practical implementations is limited, this approach efficiently leverages real-time CSI to enhance network performance while keeping computational complexity manageable.

\subsection{Step-two: Optimizing Active Beamformers}
In this stage, we assume that the shapes, $\boldsymbol{A}$, the passive reflecting elements at the RHS, $\boldsymbol{\Phi}_s$ and  
$\boldsymbol{\chi}$
are fixed to optimize the active beamforming matrix $\mathbf{W}$. Therefore, $\text{P}_2$ can be rewritten as: 
\begin{subequations} 
\begin{align} 
\text{P}_3&: 
\mathop {\max}
\limits_{\mathbf{W}} 
\sum\limits_{k = 1}^K 
\frac{\omega _k\left( {1 + \chi _k} \right)\gamma _k}
{1 + \gamma _k}
\nonumber
\\
s.t.&:~
\quad 
\eqref{p1_c1} ~{\rm and}~ \eqref{p1_c2}.
\end{align}
\end{subequations} 
It is observed that $\text{P}_3$ represents a multi-ratio fractional programming form. Utilizing the quadratic transform method as described in~\cite{10464825,8310563}, we rewrite the optimization problem $\text{P}_3$ as:
\begin{subequations} 
\begin{align} 
\text{P}_4&: 
\mathop {\max}
\limits_{\mathbf{W},\boldsymbol{\tau}} 
\sum\limits_{k = 1}^K 
2\sqrt{{\omega _k}\left( {1 + {\chi _k}} \right)}
\Re 
\left\{
\tau^*_k \sqrt{\gamma _k}
\right\}
-
|\tau_k|^2(1 + \gamma _k)
\nonumber
\\
s.t.&:~
\quad 
\eqref{p1_c1},\eqref{p1_c2},
\end{align}
\end{subequations} 
where  
$\boldsymbol{\tau}$ is the collection of auxiliary variables $ \{\tau_1,\dots,\tau_K\}$. 
With fixed $\mathbf{W}$, the optimal $\tau_k$, denoted by $\tau_k^{opt}$, can be obtained by setting the derivative of the objection of $\text{P}_4$ with respect to $\tau_k$ to zero. Thus, we have  
$\tau_k^{opt} = 
\sqrt{{\omega _k}\left( {1 + {\chi _k}} \right)}
\sqrt{\gamma _k}/
{(1 + \gamma _k)}$.
Similarly, with a fixed $\tau_k^{opt}$ and now derivative with respect to $\mathbf{W}$, we obtain 
\begin{equation}
\mathbf{w}_k^{opt} = 
\sqrt{{\omega _k}(1 + {\chi _k})}
\tau_k
\left(
\mu
\mathbf{I}_{N_\textrm{tr}}+
\sum_{k=1}^{K}
|\tau_k|^2
\mathbf{b}_k
\mathbf{b}_k^{\textrm{H}}
\right)^{-1}
\!\!\!\!
\mathbf{b}_k
\end{equation}
where 
$\mathbf{b}^{\textrm{H}}_k=
\sum\limits_{s \in \mathcal{S}} 
\mathbf{g}_{s,k}^{\textrm{H}} 
\Phi_g^{\textrm{H}}
\mathbf{H}_s$
and 
$\mu \geq 0$ is the Lagrange multiplier.

\subsection{Step-three: Optimizing Passive Beamformers}
We now assume that the shapes, $\boldsymbol{A}$, the active beamformer matrix $\mathbf{W}$ and  
$\boldsymbol{\chi}$
are fixed to optimize the passive beamforming matrix at the RHS, i.e., $\boldsymbol{\Phi}_s$. 
For notational convenience, we first revisit the SNR equation in~\eqref{SNIR}, and simplify the numerator as: 
\vspace{-6mm}
\begin{align}
\\
\nonumber
\mathbf{f}_k^{\text{H}}\mathbf{w}_k 
\!=\!
\sum\limits_{s \in \mathcal{S}} 
\mathbf{g}_{s,k}^{\textrm{H}} 
\boldsymbol{\Phi}_s^{\textrm{H}}
\mathbf{H}_s
\mathbf{w}_k
\!=\!
\sqrt \eta 
\sum\limits_{s \in \mathcal{S}} \!\!
\boldsymbol{\Theta}_s^{\textrm{H}}
\operatorname{diag} 
\left( \mathbf{g}_{s,k}^{\textrm{H}}
\right)\!
\mathbf{H}_s
\mathbf{w}_k, \forall k,
\end{align}
where $\boldsymbol{\Theta}_s = [\theta_{s,1},\dots,\theta_{s,M}]^T$.
Thus, $\text{P}_2$ can be rewritten as: 
\begin{subequations} 
\begin{align} 
\text{P}_5&: 
\mathop {\max}
\limits_{\boldsymbol{\Theta}_s} 
\sum\limits_{k = 1}^K 
\frac{
{\omega _k\left( {1 + {\chi _k}} \right)
\Big| 
\sum\limits_{s \in \mathcal{S}}
\boldsymbol{\Theta}_s^{\textrm{H}}
\mathbf{v}_{s,k}
\Big|^2
}
}
{\Big|
\sum\limits_{s \in \mathcal{S}} 
\boldsymbol{\Theta}_s^{\textrm{H}}
\mathbf{v}_{s,k}\Big|^2
+
\sigma_{\xi}^2
\Big| {\sum\limits_{s \in \mathcal{S}} 
d_{s,k}
} 
\Big|^2},
\nonumber
\\
s.t.&:
\quad 
\frac{ 
\boldsymbol{\Theta}_s^{\textrm{H}}
\mathbf{v}_{s,k}
}
{d_{s,k}}
\geq p_{s}^{thr},
\: ~~~~~
\forall s, k,
\label{p5_c1}
\\ 
&\quad\quad 
\operatorname{diag}
(
\boldsymbol{\Theta}_s
\boldsymbol{\Theta}_s^{\textrm{H}}
)
=
\mathbf{1}_{M},
\forall s,
\label{p5_c2}
\\ 
&\quad\quad 
\operatorname{rank}
(
\boldsymbol{\Theta}_s
\boldsymbol{\Theta}_s^{\textrm{H}}
)
\leq
1,
\:~~\forall s,
\label{p5_c3}
\\ 
&\quad\quad 
\boldsymbol{\Theta}_s
\boldsymbol{\Theta}_s^{\textrm{H}}
\succeq
\mathbf{0}_M,
\:~~~~~~~~\forall s,
\label{p5_c4}
\end{align}
\end{subequations} 
where 
$\mathbf{v}_{s,k}=
\sqrt \eta 
\operatorname{diag} 
\left( \mathbf{g}_{s,k}^{\textrm{H}}
\right)\!
\mathbf{H}_s
\mathbf{w}_k$, 
and 
$d_{s,k}= 
\sqrt{||\boldsymbol{a}- \boldsymbol{q}[s]||^{\kappa}
||\boldsymbol{q}[s]-\boldsymbol{u}[k]||^{\kappa}},
\forall s, k,$ with  $\boldsymbol{a},~\boldsymbol{q}[s]$ and $\boldsymbol{u}[k]$ being the AP, the $s$-th RHS, and the $k$-th UE positions, respectively. The optimization problem $\text{P}_5$ is also a multi-ratio fractional programming problem.
Constraint \eqref{p5_c2} puts a limit on the unit-unit-modulus phase shifts.
Constraints \eqref{p5_c3} and \eqref{p5_c4} are imposed to ensure $\boldsymbol{\Theta}_s
\boldsymbol{\Theta}_s^{\textrm{H}}$ holds after optimization. 
Therefore, we exploit the quadratic transform method as in the previous stage to arrive at the following transformed problem:
\begin{subequations} 
\begin{align} 
\text{P}_6&: 
\mathop {\max}
\limits_{{\boldsymbol{\Theta}}_s,\widetilde{\boldsymbol{\Theta}},\boldsymbol{\varepsilon}} 
\sum\limits_{k = 1}^K 
2\sqrt{{\omega _k}
\left( {1 + {\chi _k}} \right)}
\Re 
\left\{
\varepsilon^*_k 
\operatorname{vec}
(\widetilde{\boldsymbol{\Theta}})^{\text{H}}
\operatorname{vec}
(\widetilde{\mathbf{V}}_{k})
\right\}
\nonumber
\\
&
~~~~~~~~
-
|\varepsilon_k|^2
\left(
\left(
\operatorname{vec}
(\widetilde{\boldsymbol{\Theta}})^{\text{H}}
\operatorname{vec}
(\widetilde{\mathbf{V}}_{k})
\right)^2
+
\sigma_{\xi}^2
\Big| {\sum\limits_{s \in \mathcal{S}} 
d_{s,k}
} 
\Big|^2
\right)
\nonumber
\\
s.t.&:
\quad 
\eqref{p5_c1}-\eqref{p5_c4},
\end{align}
\end{subequations} 
where $\widetilde{\boldsymbol{\Theta}}=\big[\boldsymbol{\Theta}_1\big|\dots\big|\boldsymbol{\Theta}_S\big]$, 
$\widetilde{\mathbf{V}}_{k}=\big[\mathbf{v}_{1,k}\big|\dots\big|\mathbf{v}_{S,k}\big]$,
are constructed to simplify notations
and 
$\boldsymbol{\varepsilon}
= [\varepsilon_1,\dots,\varepsilon_K]^T$ is the auxiliary vector introduced by the quadratic transformation.  
The auxiliar variables can be updated based on 
$\varepsilon_k=
\frac{
\sqrt{\omega_k(1 + \chi _k)}
\left(
\operatorname{vec}
(\widetilde{\boldsymbol{\Theta}})^{\text{H}}
\operatorname{vec}
(\widetilde{\mathbf{V}}_{k})
\right)^2
}
{
\left(
\operatorname{vec}
(\widetilde{\boldsymbol{\Theta}})^{\text{H}}
\operatorname{vec}
(\widetilde{\mathbf{V}}_{k})
\right)^2
+
\sigma_{\xi}^2
|{\sum\limits_{s \in \mathcal{S}} 
d_{s,k}
}|^{^2}
}.
$
Finally, 
given $\boldsymbol{\varepsilon}$, the optimization of the passive beamformers can be reformulated as: 
\begin{subequations} 
\begin{align} 
\text{P}_7&: 
\mathop {\max}
\limits_{\boldsymbol{\Theta}_s,\widetilde{\boldsymbol{\Theta}}} 
-
\operatorname{vec}
(\widetilde{\boldsymbol{\Theta}})^{\text{H}}
\boldsymbol{\Upsilon}
\operatorname{vec}
(\widetilde{\boldsymbol{\Theta}})
+ 2\Re 
\left\{ 
\operatorname{vec}
(\widetilde{\boldsymbol{\Theta}})^{\text{H}}
\mathbf{z}
\right\}
\nonumber
\\
s.t.&:~
\quad 
\eqref{p5_c1},\eqref{p5_c2},\eqref{p5_c4},
\end{align}
\end{subequations} 
where 
$\boldsymbol{\Upsilon}
=
\sum_{k = 1}^K 
|\varepsilon_k|^2
\operatorname{vec}
(\widetilde{\boldsymbol{\Theta}})
\operatorname{vec}
(\widetilde{\mathbf{V}}_{k})^{\text{H}}
$
and 
$\mathbf{z}
=
\sum_{k = 1}^K 
\varepsilon_k^*
\sqrt{\omega_k(1 + \chi _k)}
\operatorname{vec}
(\widetilde{\mathbf{V}}_{k}).
$
The optimization problem $\text{P}_7$
is a quadratically constrained quadratic programming problem where the rank-one phase shift constraint~\eqref{p5_c3} is dropped to relax the problem. 
However, the solution may not be a rank-one matrix. To enforce the rank-one constraint, one can adopt a penalty-based method as described in~\cite{10464825}, yielding $\text{P}_7$ to be a convex problem that can be solved efficiently using CVX.


\section{Simulation Results}
In this study, all UE weights $\omega_k$ are set to uniform values for consistency across simulations. Based on the modeling approach in \cite{10074428}, the channel gain follows a complex Gaussian distribution, $\mathcal{CN}(0, 10^{-0.1 PL(r)})$. The path loss is defined as $PL(r) = \rho_a + 10 \rho_b \log(r) + \delta$, where $\delta$ is a Gaussian random variable with $\delta \sim \mathcal{N}(0, \sigma^2_{\delta})$. The system parameters are configured as follows: $\sigma_{\xi} = -85~\rm{[dBm]}$, 
$\rho_a = 61.4$, $\rho_b = 2$, and $\sigma_{\delta} = 5.8\rm{[dB]}$. The AP, consisting of $N_{\textrm{tr}} = 32$ antennas, is positioned at the origin, while UEs are randomly scattered within a circular region centered at $(40~\text{m}, 0~\text{m})$ with a radius of $10$ m. Additionally, we deploy $\mathcal{S} = 4$ RHS units at fixed locations, 
each having $M_x=M_y=M/200=200$ elements and two possible shapes ($|\mathcal{A}|^2=2$, $10\times6$ square-shaped and $60\times1$ ULA-shaped RHS). The total transmission power is $p_{\max} = 30$ $\rm{[dBm]}$ with $p_{s}^{thr}=20$ $\rm{[dBm]}$ \cite{10074428}.

\begin{figure}[t]
\centering
\begin{tikzpicture}
\begin{axis}[
    width=0.975\columnwidth,
    height=0.385\textwidth,
    xlabel={$p_{\max}$ [$\rm{dBm}$]},
    ylabel={Throughput [$\rm{bits/sec/Hz}$]},
    grid=major,
    legend style={
        at={(0.60,0.57)}, 
        anchor=south east, 
        font=\scriptsize,
        inner sep=0.5mm, 
        legend cell align={left},
        legend columns=1,
        /tikz/column 1/.style={
            column sep=-0pt,
        },
        /tikz/column 1/.style={
            column sep=0pt,
        },
        /tikz/row 1/.style={
            row sep=-0pt,
        },
        /tikz/row 1/.style={
            row sep=-0pt, 
        },
        /tikz/row 1/.style={
            row sep=-0pt,
        },
    },
    tick label style={font=\small},
    xlabel style={font=\footnotesize, yshift=1.35mm},
    ylabel style={font=\footnotesize, yshift=-1.5mm},
    xmin=20, xmax=45,
    xtick={20,25,30,35,40,45},
    ymin=4, ymax=18,
    ytick={4,6,8,10,12,14,16,18},
    legend entries={
        Proposed Shape-adaptive method,
        Fixed {\tikz \draw (0,0.01) rectangle (-3.1mm,1mm);}-Shaped RHS,
        ~2-bit $\Box$-Shaped RHS,
        ~2-bit {\tikz \draw (0,0.01) rectangle (-3.1mm,1mm);}-Shaped RHS,
        $~~~$ZF $\Box$-Shaped RHS,
        $~~~$ZF {\tikz \draw (0,0.01) rectangle (-3.1mm,1mm);}-Shaped RHS
    }
]

\addplot[mark=*, solid, black, line width=0.7pt] coordinates {
    (20, 6.61) (25, 7.88) (30, 9.61) (35, 11.71) (40, 14.45) (45, 17.88)
};
\addplot[mark=+, solid, blue, line width=0.7pt] coordinates {
    (20, 6.21) (25, 7.50) (30, 9.13) (35, 11.12) (40, 13.65) (45, 17.02) 
};
\addplot[mark=square*, dashdotted, red, line width=0.7pt] 
coordinates {
    (20, 5.61) (25, 6.88) (30, 8.59) (35, 10.8) (40, 12.4) (45, 14.21)
};
\addplot[mark=+, dashdotted, darkgreen, line width=0.7pt] 
coordinates {
    (20, 4.99) (25, 6.24) (30, 7.97) (35, 9.89) (40, 11.78) (45, 13.59)
};
\addplot[mark=square*, solid, lightpurple, line width=0.7pt] coordinates {
    (20, 4.69) (25, 5.37) (30, 6.11) (35, 7.03) (40, 8.39) (45, 9.91)
};
\addplot[mark=+, solid, cyan, line width=0.7pt] coordinates {
    (20, 4.51) (25, 5.12) (30, 5.90) (35, 6.81) (40, 7.92) (45, 9.22)
};


\end{axis}
\end{tikzpicture}
\vspace{-3mm}
\caption{Throughput vs. ${p}_{\max}$.}
\vspace{-5mm}
\label{fig2}
\end{figure}
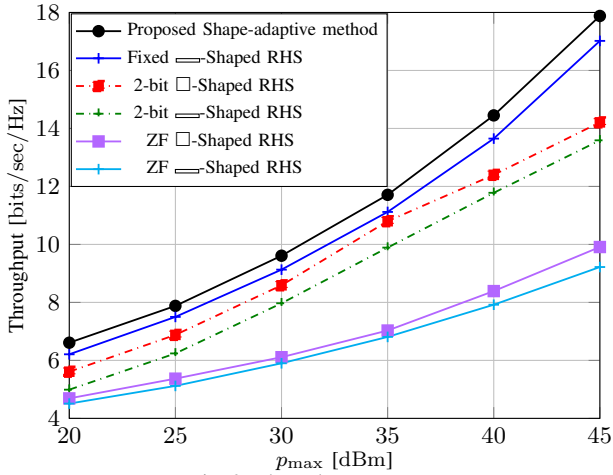

{Fig.~\ref{fig2} compares the throughput performance of the proposed shape-adaptive algorithm against three baselines:
$(i)$ the proposed algorithm applied to fixed RHS shapes (i.e., without shape adaptation),
$(ii)$ fixed-shape designs with 2-bit quantized phase shifts, and
$(iii)$ a zero-forcing (ZF) precoding scheme with randomly selected RHS phase shifts.
The shape-adaptive method achieves the highest throughput across all transmit power levels, confirming the benefit of optimizing both the RHS shape and active/passive beamformers.
The second-best performance comes from applying the same optimization algorithm to fixed RHS shapes, using ideal (continuous) phase shifts. While this setup does not benefit from shape adaptation, it outperforms the quantized and random-phase baselines due to optimized beamforming. Introducing 2-bit phase quantization leads to a noticeable drop in throughput for both RHS shapes. In this case, the square-shaped RHS outperforms the ULA-shaped RHS, owing to its larger reflective surface and broader angular coverage, which better supports multi-user signal delivery.
The ZF baseline results in the lowest throughput, underscoring the importance of intelligent RHS phase shift design and optimization.} \vspace{-0.1em}

\section{Conclusion}
In this work, we have proposed a shape-adaptive RHS deployment strategy, addressing limitations in existing RIS optimization methods. Unlike conventional approaches that focus on RIS selection, placement, or element activation, our method dynamically adjusts the shape of utilized RHS elements from a subset of predefined shapes to maximize network throughput.
We formulated a non-convex optimization problem and introduced an AO framework to optimize RHS shape adaptation, AP beamforming, and RHS phase shifts. Numerical results confirmed that shape adaptation enhances throughput, improves network robustness, and mitigates signal blockages.
We established RHS shape adaptation as a key design factor for future wireless networks. \vspace{-0.3em}

\bibliographystyle{ieeetr}
\bibliography{ref}

\end{document}